\documentclass[letterpaper, 10 pt, conference]{ieeeconf}  

\IEEEoverridecommandlockouts                              

\usepackage{amsmath} 
\usepackage{amssymb}  
\usepackage{mathtools}
\usepackage{url}
\usepackage{bbm}
\usepackage{courier}
\usepackage{algorithmicx}
\usepackage{subcaption}  
\usepackage{cleveref}
\usepackage{algorithm}
\usepackage[noend]{algpseudocode}

\newtheorem{theorem}{Theorem}
\newtheorem{corollary}{Corollary}
\newtheorem{definition}{Definition}
\usepackage{xcolor}
\usepackage{siunitx}
\usepackage[style=ieee, url=false,doi=false,eprint=false,natbib=false,mincitenames=1,maxcitenames=1,dashed=false,isbn=false, maxbibnames=99,backend=biber]{biblatex}
\newcommand{\revstephen}[1]{{\color{black} #1}}
\newcommand{\revboyd}[1]{{\color{black} #1}}

\newcommand{\eqdef}{\mathrel{\mathop=}:}

\def\baselinestretch{0.97}

\title{\LARGE \bf
Sliced Distribution Matching based on Cumulative Distribution\\ Functions
with Applications to Control
}

\author{
  Alexandros E.~Tzikas,$^{1}$ Arec Jamgochian,$^{1, 3}$ Nazim Kemal Ure,$^{1}$\\
  Mykel J. Kochenderfer,$^{1}$ and Stephen P. Boyd$^{2}$
  \thanks{$^{1}$A. E. Tzikas (corresponding author), A. Jamgochian, N. Kemal Ure, and M. J. Kochenderfer are with the Department of Aeronautics and Astronautics, Stanford University, Stanford, CA 94305, U.S.A.
        {\tt\small \{alextzik, arec, ure, mykel\}@stanford.edu}}%
    \thanks{$^{2}$S. P. Boyd is with the Department of Electrical Engineering, Stanford University, Stanford, CA 94305, U.S.A.
        {\tt\small boyd@stanford.edu}}%
    \thanks{$^{3}$TerraAI, Redwood City, CA 94063, U.S.A.}
}

\begin{document}

\maketitle
\thispagestyle{empty}
\pagestyle{empty}

\begin{abstract}
Computing the similarity between two probability distributions is a recurring theme across control. We introduce a unified family of distances between the probability distributions of two random variables that is based on the discrepancy between the cumulative distribution functions of random linear one-dimensional projections of the random variables. Our proposed distance is interpretable, computationally simple, and admits a differentiable approximation. We establish asymptotic theoretical guarantees for sample-based estimators of the distance. We empirically study the use of the distance in a two-sample test and demonstrate its ability to distinguish different distributions. Finally, we show that the distance allows for simple gradient-based solutions in control by studying distribution steering and ergodic control\footnote{The code for all presented experiments can be found at \url{https://github.com/sisl/distribution_based_control}}.
\end{abstract}



\section{Introduction}\label{sec:prior_work}

Quantifying how far two probability distributions are from one another is a central task in many control applications, such as distribution steering \cite{rapakoulias2024discrete, balci2020covariance, rapakoulias2023discrete}, ergodic control \cite{mathew2011metrics, dressel2019tutorial}, and sequential decision-making \cite{araya2010pomdp}.
In distribution steering, the Kullback–Leibler divergence \cite{rapakoulias2024discrete}, the Wasserstein distance \cite{balci2020covariance}, or the difference between the moments of the distributions \cite{rapakoulias2023discrete} is often used. In ergodic control, the discrepancy between the Fourier coefficients of the two distributions has been suggested \cite{miller2013trajectory, dressel2019tutorial}, while the distance from the simplex center has been explored in sequential problems \cite{araya2010pomdp}.

\revboyd{For the case of distributions over the reals, various notions of similarity have been proposed, such as divergences \cite{nadjahi2020statistical}, the Wasserstein metric \cite{givens1984class}, the Kolmogorov distance \cite{grossi2025refereeing}, and the continuous ranked probability score (CRPS) \cite{gneiting2007strictly}. Many of these (e.g., the first three) extend to higher dimensions, but become computationally heavy. Others (e.g., the last one) cannot be easily extended to dimensions larger than one. In the machine learning literature \cite{kolouri2018sliced, kolouri2019generalized}, projection-based methods have been proposed to compute distances between distributions in high dimensions.  These methods average discrepancies between one-dimensional projections of the distributions and therefore are typically called \textit{sliced}. Any of the aforementioned discrepancies for distributions over the reals can be used for the one-dimensional projections. To the best of our knowledge, sliced techniques are not commonly used in the area of control.} 


\begin{figure}[t]
    \centering
    \includegraphics[width=0.4\textwidth]{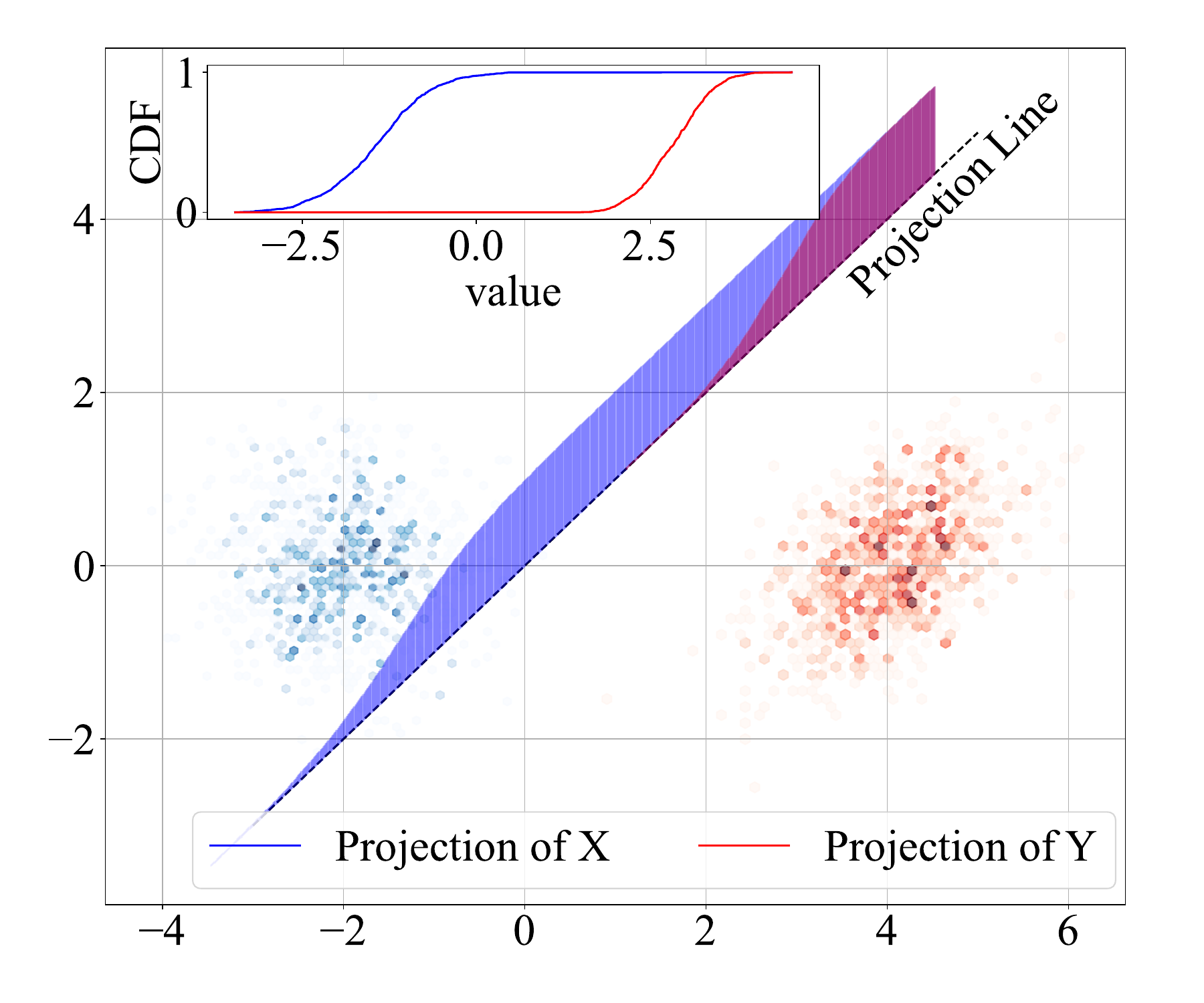}
    \caption{Overview of our proposed family of distances between the distributions of two random variables $X$ and $Y$. We average discrepancies between the CDFs of linear one-dimensional projections of the random variables. The distributions of the random variables are represented as sets of samples here. }
    \label{fig:overview}
\end{figure}

\textit{Our work unifies prior work in sliced distances and introduces a family of sliced distances between two distributions that is based on the cumulative distribution functions (CDFs) of their linear one-dimensional projections. In essence, we measure the expected dissimilarity between the CDFs of their one-dimensional linear projections. We allow for flexibility in choosing the dissimilarity function; various sliced distances are members of our proposed family, as discussed below. Any distance in our introduced family is interpretable and a differentiable estimate can be found using samples from the two distributions; we simply average (functions of) the difference in probability content over chosen half-spaces for the two distributions.} The main elements of our proposed distance are summarized in \Cref{fig:overview}. \textcite{su2015distances} proposed $L^p$ distances between CDFs for the space of one-dimensional distributions. However, they did not leverage projection-averaging to scale to higher dimensions and did not consider other distances between the CDFs. We take inspiration from distributionally robust control where linear projections of random variables are used to define interpretable and easy-to-handle constraints \cite{tzikas2025distributionally}. \textit{By providing the first (to our knowledge) proof-of-concept applications, through simple gradient-based algorithms, to distribution steering and ergodic control, we demonstrate that the introduced family of distances can be effectively used in control.}

The remainder of the paper is organized as follows. In \cref{sec:prior}, we discuss existing sliced distances. In \cref{sec:distance} we provide a unification by introducing the proposed family of distances between distributions. In \cref{sec:estimator}, we prove asymptotic results for estimators of the distance. In \cref{sec:algo}, we describe a practical algorithm to approximately compute the distance. In \cref{sec:two_sample}, we show the performance of this algorithm in a two-sample test. Finally, using the distance, in \cref{sec:distr_steer} we determine controllers for distribution steering, while in \cref{sec:ergodic} we compute a control sequence for ergodic control.

\section{Prior Work on Sliced Distances}\label{sec:prior}

We discuss various projection-based strategies in turn and compare them against our proposed family of distances. Applications of such methods span generative modeling, clustering, and barycenter computation, where full multivariate optimal transport is infeasible \cite{deshpande2018generative, kolouri2018sliced}.

The sliced Wasserstein (SW) distance makes multivariate Wasserstein distances tractable by projecting the distributions onto one-dimensional subspaces and averaging their Wasserstein distances. The Wasserstein distance can be calculated in closed form in one dimension and involves the integral of the pointwise absolute difference of the CDFs \cite{rabin2011wasserstein}. Building on this idea, \textcite{kolouri2019generalized} proposed generalized sliced Wasserstein (GSW) distances using nonlinear projections, as well as max-sliced variants that optimize, rather than average, over projections. Projection-robust sliced Wasserstein distances that seek maximally discriminative subspaces have also ben suggested \cite{lin2021projection}. \textit{The SW is included in our introduced family of distances.}

The Gromov–Wasserstein distance looks at the minimum expected absolute difference of distances between points generated by each distribution. The minimum is over the coupling of the distributions, setting the two distributions as the marginals. In the one-dimensional case, under assumptions, the distance can be computed by solving a quadratic assignment problem. By performing random one-dimensional projections of the two distributions and averaging their Gromov–Wasserstein distance we obtain the sliced variant \cite{titouan2019sliced}. The sliced Sinkhorn distance relies on the entropy-regularized optimal transport cost and can be computed by solving the dual of a convex program \cite{nadjahi2020statistical}. \textit{Our distance does not require solving any inner optimization problem.}

Closely related to our proposed family of sliced distances are two statistics used for multivariate two-sample tests; the projection-averaged Cramer-von-Mises statistic and the projection-averaged Kolmogorov-Smirnov (KS) statistic. The former applies the discrepancy function $d(f,g) = \int_{-\infty}^{\infty} (f(x)-g(x))^2 dx$ to the CDFs of the random one-dimensional projections and is shown to have an efficient closed-form solution by~\textcite{kim2020robust} and by~\textcite{li2020projective}. The latter  applies the discrepancy function $d(f,g) = \sup_{x} |f(x)-g(x)|$ to the CDFs of the random one-dimensional projections \cite{grossi2025refereeing}.

\section{The Family of Distances between Distributions}\label{sec:distance}

Suppose $X: \Omega_X \rightarrow \mathbb{R}^n$ and $Y:\Omega_Y \rightarrow \mathbb{R}^n$ are two random variables, defined over potentially different probability spaces $\left(\Omega_X, \mathcal{F}_X, \mathbb{P}_X \right)$ and $\left(\Omega_Y, \mathcal{F}_Y, \mathbb{P}_Y \right)$, respectively.
We introduce a class of distances between the laws (or induced probability measures on $\mathbb{R}^n$) of $X$ and $Y$ that is based on linear one-dimensional projections of $X$ and $Y$.

We first define the distance for the set $\mathcal{C}$ of univariate CDFs $f: \mathbb{R} \rightarrow \left[0, 1\right]$. 
\begin{definition}\label{def:distance_func}
(\textbf{Distance in $\mathcal{C}$}). A function $d: \mathcal{C} \times \mathcal{C} \rightarrow \mathbb{R} \cup \{\infty\}$ is a distance for the set $\mathcal{C}$ if and only if it satisfies the following axioms:
\begin{enumerate}
    \item \textit{Non-negativity}: $d(f,g) \geq 0$,
    \item \textit{Symmetry}: $d(f,g)=d(g,f)$,
    \item \textit{Triangle inequality}: $d(f,g) \leq d(f,h) + d(h, g)$,
    \item \textit{Identity of indiscernibles}: $d(f,g) = 0 \Leftrightarrow f(x) = g(x)\ \mathrm{for\ all}\ x$.
\end{enumerate}
\end{definition}

Next, we define the one-dimensional linear projection of an $n$-dimensional random variable.
\begin{definition}
    (\textbf{One-dimensional linear projection of a random variable}). Let $\Tilde{q} \in \mathbb{R}^n$ be a vector. Then, $X_{\Tilde{q}}:=\Tilde{q}^\top X$ is a linear one-dimensional projection of $X$.
\end{definition}

By considering $X_{\Tilde{q}}$ as a weighted sum of the coordinates of $X$, it is obvious that $X_{\Tilde{q}}$ is a random variable over $\left(\Omega_X, \mathcal{F}_X, \mathbb{P}_X \right)$, as a Borel function of $X$. Because $X_{\Tilde{q}}$ is a one-dimensional random variable, it admits a CDF, $F_X^{\Tilde{q}}(t):=\mathbb{P}\left( \Tilde{q}^\top X \leq t \right)$.

We now define our measure of similarity between the laws (distributions) of $X$ and $Y$. 
\begin{definition}
     (\textbf{Measure of similarity between distributions}). Suppose $\mathcal{P}_q:=(\Omega_q, \mathcal{F}_q, \mathbb{P}_q)$ is a probability space and $q: \Omega_q \rightarrow \mathbb{S}^n$ is a random variable over $\mathcal{P}_q$, where $\mathbb{S}^n$ is the unit sphere in $n$ dimensions.
     Suppose $d$ is a distance in $\mathcal{C}$. Letting $\omega' \in \Omega_q$, we consider the random variable $d(F_X^{q(\omega')}, F_Y^{q(\omega')})$ over $\mathcal{P}_q$. We define the following measure of similarity between the laws of random variables $X$ and $Y$, respectively:
     \begin{equation}\label{eq:dist_main}
         \Delta(X, Y) := \mathbb{E}_{q}\left[ d(F_X^{q}, F_Y^{q})\right].
     \end{equation}
\end{definition}

Note that the random variable $q$ takes values only on the unit sphere. This is a choice justified by the fact that $F_X^q(t) = F_X^{\frac{q}{\lVert q \rVert}}(t / \lVert q \rVert)$ for $q \neq 0$. Moreover, although presented in the context of $\mathbb{R}^n$, $\Delta(X,Y)$ is also applicable for distributions over matrices because matrices can be equivalently viewed as vectors (by stacking their columns or rows). 
\revboyd{It is important to note that $\Delta(X,Y)$ depends on the distribution of the random variable $q$ on the unit sphere. A common choice is to use a uniformly distributed $q$, but, as we will see later on, this is not the only one.}

We now prove that $\Delta(X,Y)$ is a distance in the space of distributions.
\begin{theorem}
    (\textbf{Distance property}). $\Delta(X,Y)$, as given in \eqref{eq:dist_main}, is a distance in the space of distributions.
\end{theorem}
\begin{proof}
The non-negativity property of $\Delta$ follows from the non-negativity property of $d$. Similarly, the symmetry property of $\Delta$ follows from the symmetry property of $d$. 

We now prove the triangle inequality. Consider an additional random variable $Z: \Omega_Z \rightarrow \mathbb{R}^n$ over $\left(\Omega_Z, \mathcal{F}_Z, \mathbb{P}_Z \right)$. Then:
\begin{equation}\label{eq:triangle_ineq}
\begin{aligned}
    \Delta(X,Y) &:= \mathbb{E}_{q}\left[ d(F_X^{q}, F_Y^{q})\right] \\
    & \leq \mathbb{E}_{q} \left[ d(F_X^{q}, F_Z^{q}) + d(F_Z^{q}, F_Y^{q})\right] \\
    & = \Delta(X,Z) + \Delta(Z, Y).
\end{aligned}
\end{equation}

Finally, we prove the identity of indiscernibles.

($\Rightarrow$). Suppose $X$ and $Y$ are equal in distribution: $X\stackrel{D}{=}Y$. Then, for all $\Tilde{q} \in \mathbb{S}^n$, it holds that $F_X^{\Tilde{q}} = F_Y^{\Tilde{q}}$ pointwise. This implies that $d (F_X^{\Tilde{q}}, F_Y^{\Tilde{q}})=0$ for all $\Tilde{q} \in \mathbb{S}^n$, by the zero identity of $d$. Therefore
    $\mathbb{E}_{q}\left[ d(F_X^{q}, F_Y^{q})\right]=0$ and we obtain $\Delta(X,Y)=0$.
    
($\Leftarrow$). Suppose $\Delta(X,Y)=0$. Then $\mathbb{E}_{q}\left[ d(F_X^{q}, F_Y^{q})\right]=0$. By Markov's inequality, for any $\alpha >0$
    \begin{equation}
        \mathbb{P}_q\left(d(F_X^{q}, F_Y^{q}) \geq \alpha\right) \leq \dfrac{\mathbb{E}_{q}\left[ d(F_X^{q}, F_Y^{q})\right]}{\alpha}=0. 
    \end{equation}
    Therefore $d(F_X^{q}, F_Y^{q}) =0 $ with probability $1$, which implies that $F_X^{q} = F_Y^{q}$ pointwise with probability $1$. Assuming the support of $q$ is all $\mathbb{S}^n$, we obtain that for all $\Tilde{q} \in \mathbb{S}^n$: $\Tilde{q}^\top X \stackrel{D}{=}\Tilde{q}^\top Y $, which implies that for all $\Tilde{q}\in \mathbb{R}^n$: $\Tilde{q}^\top X \stackrel{D}{=}\Tilde{q}^\top Y $. By applying the Cramér--Wold theorem \cite{measure}, it follows that $X \stackrel{D}{=}Y$. Intuitively, if the two random vectors $X$ and $Y$ agree on each projection, then they have the same distribution.
\end{proof}

\subsection*{Specific Cases}
Any function $d$ that satisfies the requirements of \Cref{def:distance_func} can be used in $\Delta(X,Y)$. Here, we outline a few choices.
\begin{itemize}
    \item Suppose $d(f,g)=\sup_x |f(x)-g(x)|$. In this case, $\Delta(X,Y)$ is the expected KS distance between linear one-dimensional projections of $X$ and $Y$. This choice of $d$ can be particularly helpful, if we care about the maximum disagreement between the CDFs and not where the disagreement occurs.
    \item Suppose that $d(f,g) = \int_{-\infty}^\infty |f(x)-g(x)| dx$. Then, $\Delta(X,Y)$ captures the expected total difference between the CDFs of linear one-dimensional projections of $X$ and $Y$. This is equivalent to the 1-Wasserstein distance on distributions. Note that we can generalize to $d(f,g) = \int_{-\infty}^{\infty} c(x)|f(x)-g(x)|dx$ for function $c$ such that $c(x) >0$ for all $x$. This can be helpful, if there are specific values of $x$ for which we care about the agreement of the two CDFs. 
\end{itemize}

\section{A Sample-based Estimator}\label{sec:estimator}
We define a sample-based estimator of $\Delta(X,Y)$. For simplicity we will assume that $F_Y^{\Tilde{q}}$ is known for all $\Tilde{q}$. This can occur, for example, if we wish to compute the distance of the distribution of $X$ to a known distribution; the distribution of $Y$.

Suppose we are given the samples (represented as random variables) $\left(q_{i}, X_{i, 1}, \ldots, X_{i, N} \right)_{i=1}^H$ in a probability space with probability measure $\mathbb{Q}=\mathbb{P}_q \times (\times_{k=1}^N \mathbb{P}_X )$. The inner sample size $N$ will be used to estimate CDFs, while the outer sample size $H$ will be used to estimate the expectation. The samples $q_{i}$ are assumed independent and identically distributed to $q$. Similarly, the samples $X_{i,j}$ are assumed independent and identically distributed to $X$. Further, $\left(q_{i}\right)_{i=1}^H \perp \left( X_{i, 1}, \ldots, X_{i, N}\right)_{i=1}^H$, where $\perp$ denotes independence. 

For each $i$, we can obtain an estimator of $F_X^{q_{i}, N}$ by calculating the empirical CDF of the samples $(X_{i,1}, \dots, X_{i, N})$:
\begin{equation}
    \hat{F}_X^{q_{i}, N}(t) \eqdef \dfrac{1}{N}\sum_{j=1}^N \mathbbm{1}{\left( q_{i}^\top X_{i, j} \leq t\right)}.
\end{equation}
We drop the dependence on $N$ for simplicity below.
We can now define an estimator for $\Delta(X,Y)$ as follows.
\begin{definition}
    (\textbf{Estimator of $\Delta(X,Y)$}). Let $\hat{\Delta}(X,Y)$ be the following sample-based estimator of $\Delta(X,Y)$:
    \begin{equation}
        \hat{\Delta}(X,Y) \eqdef \dfrac{1}{H} \sum_{i=1}^H d(\hat{F}_X^{q_{i}}, F_Y^{q_{i}}).
    \end{equation}
\end{definition}

We study the asymptotic behavior of the estimator. First, we determine the asymptotic convergence of its mean and variance, and then provide an approximation for its asymptotic distribution.
\begin{theorem}\label{thm:asymptotic_conv_mean}
    (\textbf{Asymptotic mean of estimator}). Let $g: \mathcal{C} \rightarrow \mathbb{R}$. Suppose $d(\cdot, g): \mathcal{C} \rightarrow \mathbb{R}$ is continuous with respect to the sup-norm in its first argument and bounded above uniformly by a constant $M$ for all $g$. Then, the estimator $\hat{\Delta}(X,Y)$ is asymptotically unbiased, as $N \rightarrow \infty$.
\end{theorem}
\begin{proof}
    Obviously:
    \begin{equation*}
        \mathbb{E}\left[ \hat{\Delta}(X,Y)\right] = \mathbb{E}\left[ d(\hat{F}_X^{q_{i}}, F_Y^{q_{i}})\right].
    \end{equation*} 
    By the tower property of expectation
    \begin{equation}
    \begin{aligned}
        \mathbb{E}\left[ d(\hat{F}_X^{q_{i}}, F_Y^{q_{i}})\right] = \mathbb{E}\left[ \mathbb{E}\left[ d(\hat{F}_X^{q_{i}}, F_Y^{q_{i}}) \mid q_{i}\right]\right],
    \end{aligned}
    \end{equation}
    and, in turn, by the independence assumptions, for any realization $\Tilde{q}$ of $q$,
    \begin{equation}\label{eq:tower}
        \mathbb{E}\left[ d(\hat{F}_X^{q_{i}}, F_Y^{q_{i}}) \mid q_{i}=\Tilde{q}\right] = \mathbb{E}_\mathbb{P}\left[ d(\hat{F}_X^{\Tilde{q}}, F_Y^{\Tilde{q}})\right],
    \end{equation}
    where $\mathbb{P} = \times_{i=1}^N \mathbb{P}_X$. Eq.~\eqref{eq:tower} states that the conditional expectation of the distance for a given realization of $q$ is equal to the expected distance for the given linear one-dimensional projection.
    
    By the Glivenko–-Cantelli theorem,
    \begin{equation}
        \mathbb{P}\left(\mathrm{sup}_x |\hat{F}_X^{\Tilde{q}}(x) - F_X^{\Tilde{q}}(x)|\rightarrow0\right)=1,
    \end{equation}
    where the limit is for $N \rightarrow \infty$.
    The continuity assumption of $d$ implies that 
    \begin{equation}
    \begin{aligned}
        1 &= \mathbb{P}\left( \mathrm{sup}_x |\hat{F}_X^{\Tilde{q}}(x) - F_X^{\Tilde{q}}(x)|\rightarrow0\right)\\
        &\leq \mathbb{P}\left( d(\hat{F}_X^{\Tilde{q}}, F_Y^{\Tilde{q}}) \rightarrow d(F_X^{\Tilde{q}}, F_Y^{\Tilde{q}})\right) 
    \end{aligned}
    \end{equation}
    and therefore $d(\hat{F}_X^{\Tilde{q}}, F_Y^{\Tilde{q}}) \stackrel{a.s.}{\longrightarrow} d(F_X^{\Tilde{q}}, F_Y^{\Tilde{q}})$ as $N \rightarrow \infty$. Because $d$ is bounded, we apply the dominated convergence theorem
    to get 
    \begin{equation}
        \mathbb{E}_\mathbb{P}\left[ d(\hat{F}_X^{\Tilde{q}}, F_Y^{\Tilde{q}}) \right] \rightarrow  d(F_X^{\Tilde{q}}, F_Y^{\Tilde{q}})\ \mathrm{as}\ N \rightarrow \infty.
    \end{equation}
    This result holds for every realization $\Tilde{q}$.
    
    The boundedness of $d$ implies that $\mathbb{E}_\mathbb{P}\left[ d(\hat{F}_X^{\Tilde{q}}, F_Y^{\Tilde{q}}) \right]$ is also bounded by the same constant for any realization $\Tilde{q}$. Therefore, the random variable $\mathbb{E}\left[ d(\hat{F}_X^{q_{i}}, F_Y^{q_{i}}) \mid q_{i}\right]$ is bounded and, by applying the dominated convergence theorem to it, we get
    \begin{equation}
    \mathbb{E}\left[ \hat{\Delta}(X,Y)\right]  \rightarrow \mathbb{E}_q\left[ d(F_X^{q}, F_Y^{q})\right]= \Delta(X,Y)
    \end{equation}
    as $N \rightarrow \infty$.
    \end{proof}
    
Various choices of $d$ satisfy the continuity condition required by \Cref{thm:asymptotic_conv_mean}. For example, suppose $d(f,g) = \max_x |f(x)-g(x)|$. Consider the sequence of functions $f_n(x)$ such that $\max_x |f_n(x) - f(x)| \rightarrow 0$. Then:
\begin{equation*}
    \begin{aligned}
    &d(f_n, g) \leq \max_x |f_n(x)-f(x)| + \max_x |f(x)-g(x)|,\\
    &d(f, g) \leq  \max_x |f_n(x)-f(x)| + \max_x |f_n(x)-g(x)|,
    \end{aligned}
\end{equation*}
and we obtain
\begin{equation*}
    \lim_{n \rightarrow \infty} d(f_n, g) \leq d(f,g), \quad \lim_{n \rightarrow \infty} d(f_n, g) \geq d(f,g),
\end{equation*}
i.e., $\lim_{n \rightarrow \infty} d(f_n, g) = d(f,g)$. The same result holds for $d(f,g) = \int_{-\infty}^{\infty} c(x)|f(x)-g(x)|dx$, where $c$ is a probability density function (PDF). To see this, the reader can follow the steps outlined above.

Furthermore, because CDFs are constrained between $0$ and $1$, we see that both $d(f,g) = \max_x |f(x)-g(x)|$ and $d(f,g) = \int_{-\infty}^{\infty} c(x)|f(x)-g(x)|dx$ satisfy $d(f,g) \leq 1$, i.e., are bounded.

\begin{corollary}\label{cor:asymptotic_cov_var}
(\textbf{Asymptotic variance of estimator}). Let $g: \mathcal{C} \rightarrow \mathbb{R}$. Suppose $d(\cdot, g): \mathcal{C} \rightarrow \mathbb{R}$ is continuous with respect to the sup-norm in its first argument and bounded above uniformly by a constant $M$ for all $g$. Then, 
\begin{equation}
    \mathrm{var}\left[ \hat{\Delta}(X,Y) \right] \rightarrow \frac{1}{H}
    \mathrm{var}_q\left[ d(F_X^q, F_Y^q)\right]\ \mathrm{as}\ N \rightarrow \infty.
\end{equation}  
\end{corollary}

\begin{proof}
    Obviously
    \begin{equation}\mathrm{var}\left[ \hat{\Delta}(X,Y)\right] = \dfrac{1}{H} \mathrm{var}\left[ d(\hat{F}_X^{q_{i}}, F_Y^{q_{i}})\right]
    \end{equation}
    and
    \begin{equation}
    \begin{aligned}
        &\mathrm{var}\left[ d(\hat{F}_X^{q_{i}}, F_Y^{q_{i}})\right] = \\
        &\mathbb{E}\left[ d^2(\hat{F}_X^{q_{i}}, F_Y^{q_{i}})\right] - \mathbb{E}\left[ d(\hat{F}_X^{q_{i}}, F_Y^{q_{i}})\right]^2.
    \end{aligned}
    \end{equation}
    We study the asymptotic behavior of the two terms on the right-hand side in turn. By a similar argument as in \Cref{thm:asymptotic_conv_mean}, we obtain
    \begin{equation}
        \mathbb{P}\left(d^2(\hat{F}_X^{\Tilde{q}}, F_Y^{\Tilde{q}}) \rightarrow d^2(F_X^{\Tilde{q}}, F_Y^{\Tilde{q}})\right) = 1 
    \end{equation}
    for all realizations $\Tilde{q}$. Because $d$ is bounded, by applying the dominated convergence theorem, we get
    \begin{equation}
        \mathbb{E}_\mathbb{P}\left[ d^2(\hat{F}_X^{\Tilde{q}}, F_Y^{\Tilde{q}}) \right] \rightarrow  d^2(F_X^{\Tilde{q}}, F_Y^{\Tilde{q}})\ \mathrm{as}\ N \rightarrow \infty.
    \end{equation}
    A second application of the dominated convergence theorem, this time for the sequence of random variables $\mathbb{E}\left[ d^2(\hat{F}_X^{q_{i}}, F_Y^{q_{i}}) \mid q_i\right]$, gives
    \begin{equation}\label{eq:mean_square}
        \mathbb{E}\left[ d^2(\hat{F}_X^{q_{i}}, F_Y^{q_{i}})\right] \rightarrow \mathbb{E}_q \left[ d^2(F_X^{q}, F_Y^{q})\right]\ \mathrm{as}\ N \rightarrow \infty.
    \end{equation}
    By \Cref{thm:asymptotic_conv_mean}, we directly get
    \begin{equation}\label{eq:square_mean}
        \mathbb{E}\left[ d(\hat{F}_X^{q_{i}}, F_Y^{q_{i}})\right]^2 \rightarrow \Delta^2(X,Y)\ \mathrm{as}\ N \rightarrow \infty.
    \end{equation}
    By combining \eqref{eq:mean_square} and \eqref{eq:square_mean}, we get the result.
\end{proof}

\Cref{thm:asymptotic_conv_mean} and \Cref{cor:asymptotic_cov_var} provide the asymptotic mean and variance of $\hat{\Delta}(X,Y)$, respectively. We can also determine the asymptotic distribution of $\hat{\Delta}(X,Y)$.

\begin{theorem}\label{thm:clt}
(\textbf{Asymptotic distribution of estimator}). Let $g: \mathcal{C} \rightarrow \mathbb{R}$. Suppose $d(\cdot, g): \mathcal{C} \rightarrow \mathbb{R}$ is continuous with respect to the sup-norm in its first argument and bounded above uniformly by a constant $M$ for all $g$. Then,  
\begin{equation}
        \sqrt{H}\left(\hat{\Delta}(X,Y) - \Delta(X,Y)\right) \stackrel{D}{\longrightarrow} \mathcal{N}\left(0, 
    \mathrm{var}_q\left[ d(F_X^q, F_Y^q)\right] \right),
    \end{equation}
    as $H, N \rightarrow \infty$.
\end{theorem}
\begin{proof}
    Suppose $N_H$ is a strictly increasing function of $H$. Consider the triangular array of random variables 
    \[
\begin{array}{ccccc}
Z_1^{(1)} &        &        &       & \\
Z_2^{(1)} & Z_s^{(2)} &        &   &     \\
\vdots & \vdots & \ddots &        &\\
Z_H^{(1)} & Z_H^{(2)} & \dots & Z_H^{(H)} &\\
\vdots  & \vdots  & \vdots  & \vdots & \ddots \\
\end{array},
\]
    where
    \begin{equation}
    \begin{aligned}
        &\left( Z_H^{(i)}\right)_{i=1}^{H}\eqdef\\
        &\left( d(\hat{F}_X^{q_{i}, N_H}, F_Y^{q_{i}}) - \mathbb{E}\left[ d(\hat{F}_X^{q_{i}, N_H}, F_Y^{q_{i}}) \right]\right)_{i=1}^H,
    \end{aligned}
    \end{equation} and each $\hat{F}_X^{q_{i}, N_H}$ depends on $\left(X_{i, k} \right)_{k=1}^{N_H}$. 
    Obviously
    \begin{equation*}
        \mathbb{E}\left[Z_H^{(i)} \right] = 0,\quad \mathrm{var}\left[ Z_H^{(i)}\right] = \mathrm{var}\left[ d(\hat{F}_X^{q_{1}, N_H}, F_Y^{q_{1}})\right].
    \end{equation*}
    By \Cref{cor:asymptotic_cov_var}, $\mathrm{var}\left[ Z_H^{(i)}\right] \rightarrow \mathrm{var}_q\left[ d(F_X^q, F_Y^q)\right]$ as $H \rightarrow \infty$. We define
    \begin{equation*}
        \bar{Z}^{(H)}:= \dfrac{1}{H}\sum_{i=1}^H Z_i^{(H)}. 
    \end{equation*}

    Because the distance $d$ is bounded, we can easily confirm that the Lindeberg condition 
    \begin{equation}
        \dfrac{1}{H} \sum_{i=1}^H \mathbb{E}\left[ \left(Z_i^{(H)} \right)^2\mathbbm{1}{\left(|Z_i^{(H)}|\geq \epsilon \sqrt{H}\right)}\right] \rightarrow 0,\ \mathrm{as}\ H \rightarrow \infty
    \end{equation}
    holds for all $\epsilon >0$. By applying the Lindeberg central limit theorem \cite{billingsley2017probability}, we obtain
    \begin{equation}
        \sqrt{H}\bar{Z}^{(H)} \stackrel{D}{\longrightarrow} \mathcal{N}\left(0, \mathrm{var}_q\left[ d(F_X^q, F_Y^q)\right]\right),\ \mathrm{as}\ H \rightarrow \infty. 
    \end{equation}
    This implies the desired result.
\end{proof}

\Cref{thm:clt} proves the intuitive fact that an estimator based on empirical averages is asymptotically normally distributed. It also allows us to better understand the behavior of $\hat{\Delta}(X,Y)$. Specifically, suppose $\Delta(X,Y) = \alpha$, where $\alpha \in \left[0, 1\right]$. By approximating $d(F_X^q,F_Y^q)$ as a Bernoulli random variable with probability of success $\alpha$, we get that $\mathrm{var}_q\left[ d(F_X^q, F_Y^q)\right] = \alpha (1-\alpha)$. Therefore, the asymptotic signal-to-noise ratio for the estimator is
$\sqrt{H\alpha/(1-\alpha)}$. This is an  increasing function of $\alpha$. This, in turn, means that it is easier to conclude that two distributions are different than to conclude that two distributions are similar. Further, by increasing the number of samples $H$, we can reduce the standard error as desired.

\section{A Practical Algorithm}\label{sec:algo}

We suggest a practical way to approximately compute $\Delta(X,Y)$ for a specific choice of the function $d$. For this choice of $d$ we argue that $\Delta(X,Y)$ is easily computable and interpretable. We will assume  
\begin{equation}\label{eq:dist_heur}
    d(F_X^q, F_Y^q)\eqdef \int_{-\infty}^{\infty} c(x\mid q)\lvert F_X^q(x) - F_Y^q(x) \rvert dx,
\end{equation}
where $c(x\mid q)$ is a PDF. 
We can write 
\begin{equation}
\begin{aligned}
    &d(F_X^q, F_Y^q)\\
    &=\int_{-\infty}^{\infty} c(x\mid q)\lvert\mathbb{P}\left(q^\top X \geq x \right) - \mathbb{P}\left(q^\top Y \geq x \right) \rvert dx
\end{aligned}
\end{equation}
and interpret $d$ as the expected (with respect to PDF $c(\cdot \vert q)$) absolute difference in probability content of $X$ and $Y$ on half-spaces with normal vector $q$. In this case, we can write 
\begin{equation}\label{eq:interpret}
\begin{aligned}
    &\Delta(X,Y) = \\
    &\int_{\mathbb{S}^n} \int_{-\infty}^\infty p_q(\Tilde{q}) c(x \mid \Tilde{q})
\lvert\mathbb{P}\left(\Tilde{q}^\top X \geq x \right) - \mathbb{P}\left(\Tilde{q}^\top Y \geq x \right) \rvert dx d\Tilde{q},
\end{aligned}
\end{equation}
where $p_q$ is a PDF over $\mathbb{S}^n$. \revstephen{This distance shares similarities with the sliced Wasserstein distances \cite{kolouri2019generalized} and the continuous ranked probability score \cite{gneiting2007strictly}.}

Eq.~\eqref{eq:interpret} shows that $\Delta(X,Y)$ effectively 
samples random half-spaces and computes the difference in probability within each half-space for the two random variables $X$ and $Y$. Therefore $\Delta(X,Y)$ can be interpreted as the expected difference in probability content within a half-space for the two distributions. 

Because $\lvert\mathbb{P}\left(\Tilde{q}^\top X \geq x \right) - \mathbb{P}\left(\Tilde{q}^\top Y \geq x \right) \rvert \in \left[ 0, 1\right]$, we notice that $\Delta(X,Y)$ is constrained between $0$ and $1$.
Although, we will not discuss it further, this makes $\Delta(X,Y)$ ideal as a reward for sequential decision-making algorithms like Monte Carlo tree search (MCTS), where scaling the cost can significantly improve performance \cite{kochenderfer2022algorithms}.

In various applications one of the random variables, $X$ or $Y$, is represented as a set of samples, while the PDF of the other is given explicitly. This is common, for example, in distribution steering, where the target distribution is given as a PDF, while the the current state is given as a set of samples. Suppose that the PDF of $Y$ is given as a Gaussian mixture model (GMM). 
The GMM is a universal approximator, in the sense that any smooth density can be approximated with arbitrary error by a GMM \cite{bengio2017deep}. We show that $\mathbb{P}\left(q^\top Y + b  \geq 0\right)$ can be computed in closed form. This holds because 
\begin{equation}\label{eq:analyt_dist}
    \mathbb{E}_{y \sim \mathcal{N}(\bar{y}, \Sigma)} \mathbbm{1}{\left(q^T y + b  \geq 0\right)} = \mathbb{P}\left( r \geq 0 \right),
\end{equation}
where $r \sim \mathcal{N}(q^T \bar{y} + b, q^T \Sigma q)$ is a one-dimensional Gaussian. The last probability can be easily computed. Therefore, if $Y$ follows a  GMM density $\mu_Y$, we can compute $\mathbb{E}_{y \sim \mu_Y} \mathbbm{1}{\left(q^T y + b  \geq 0\right)}$ as follows: we weight the probabilities of the form (\ref{eq:analyt_dist}) for each component of the GMM by their corresponding weights and add the results.

Assuming $X$ is given as a set of samples $\lbrace x^{i} \rbrace_{i=1}^N$ and the PDF of $Y$, $\mu_Y$, is known explicitly, we show how to approximate $\Delta(X,Y)$ in \Cref{alg:dist_heur} \revboyd{with complexity $\mathcal{O}\left(HN(n + n_\mathrm{values})\right)$. }\revboyd{We denote the output of \Cref{alg:dist_heur} as $D_{H, n\mathrm{values}}(X,Y)$. } 

\begin{algorithm}
\caption{Approximation of $\Delta(X, Y)$ with $d$ as in \eqref{eq:dist_heur}}\label{alg:dist_heur}
\begin{algorithmic}[1]
\State \textbf{Inputs:} density $\mu_Y$ of $Y$, samples $\lbrace x_{i} \rbrace_{i=1}^N$ of $X$
\State \textbf{Parameters:}
$H$, $n_\mathrm{values}$,
\State $\Delta \leftarrow 0$
\For{$k = 1, \dots, H$}
    \State get a half-space $q_{k} \in \mathbb{S}^n$ according to the density $p_q(\cdot)$
    \For{$j = 1, \dots, n_\mathrm{values}$}
        \State get a value $b_{j} \in \mathbb{R}$ according to the density $c(\cdot \mid q_{k})$
        \State $p_{\mathrm{samples}} \leftarrow \dfrac{1}{N} \sum_{i=1}^N \mathbbm{1}\left(q_{k}^\top x_{i} \geq b_{j}\right)$
        \State $p_{\mu_Y} \leftarrow \mathbb{E}_{y \sim \mu_Y} \mathbbm{1}{\left(q_{k}^\top y \geq b_{j}\right)} $
        \State $\Delta \leftarrow \Delta + \dfrac{1}{H n_\mathrm{values}}\lvert p_\mathrm{samples} - p_{\mu_Y} \rvert$
    \EndFor
\EndFor
\State \Return $\Delta$
\end{algorithmic}
\end{algorithm}

An important aspect in \Cref{alg:dist_heur} is the procedure for obtaining the vectors $q_{k}$ and the values $b_{j}$. \revstephen{The approximate value for $\Delta(X,Y)$, $D_{H, n\mathrm{values}}(X,Y)$,  depends on this choice.}
Although \Cref{alg:dist_heur} works for arbitrary densities $p_q(\cdot)$ and $c(\cdot \vert q)$, uniformly sampling the directions $q_{k}$ from the unit sphere and picking the $b_{j}$ as the quantiles of the sample set $\{ q_k^\top x_{i}\}_{i=1}^N$ worked well in practice. 

The modifications necessary to \Cref{alg:dist_heur} in the case that $Y$ is represented as a set of samples or the PDF of $X$ is given explicitly are obvious.

\subsection{Differentiability Considerations}
In many applications, it is important for the estimate $D_{H, n\mathrm{values}}(X,Y)$ to be differentiable with respect to the samples $\{ x_{i} \}_{i=1}^N$. For example, in distribution steering we want to find a controller that steers the distribution of the samples $\{ x_{i} \}_{i=1}^N$ towards a target distribution $\mu_Y$. The estimate provided by \Cref{alg:dist_heur} is not differentiable, because the indicator function in line 8 is not differentiable.
We can obtain a differentiable modification of the estimate by replacing $\mathbbm{1}(q_{k}^\top x_{i} \geq b_{j})$ with $\sigma\left(V(q_{k}^\top x_{i} - b_{j})\right)$, where $\sigma$ is the sigmoid function and $V$ is a large scalar. We pick $V=100$, which makes the resulting function steep near the origin. \revboyd{Below, when taking the gradient of $D_{H, n\mathrm{values}}(X,Y)$ we assume the sigmoid approximation is used. The gradient of $D_{H, n\mathrm{values}}(X,Y)$ is a stochastic gradient of $\Delta(X,Y)$. }


\section{The Two-Sample Test}\label{sec:two_sample}

An important aspect of the distance estimate $D_{H, n\mathrm{values}}(X,Y)$ by \Cref{alg:dist_heur} is its ability to distinguish sets of samples produced by different distributions and identify sets of samples produced by the same distribution. We demonstrate the effectiveness of $D_{H, n\mathrm{values}}(X,Y)$ empirically with a two-sample test. We consider two sets of samples $\{ x_i\}_{i=1}^N$ and $\{ y_i\}_{i=1}^N$ in $\mathbb{R}^n$ produced by distributions $\mu_X$ and $\mu_Y$, respectively, and the two cases:
\begin{equation}
    \mathrm{H_0}: \mu_X = \mu_Y,\quad \mathrm{H_A}: \mu_X \neq \mu_Y.
\end{equation}

We look at the empirical distribution of $D_{H, n\mathrm{values}}(X,Y)$ in the null $\mathrm{H_0}$ and alternative $\mathrm{H_A}$ cases for $n\mathrm{\ =2}$ and $n\mathrm{ =50}$. The randomness in $D_{H, n\mathrm{values}}(X,Y)$ comes from the drawn samples, $\{ x_i\}_{i=1}^N$ and $\{ y_i\}_{i=1}^N$, and the half-spaces in \Cref{alg:dist_heur}. We set $N\mathrm{ =1000}$, $H\mathrm{ =300}$ and $n_\mathrm{values}\mathrm{ =100}$. In the null case $\mathrm{H_0}$ the samples are produced from $\mu_X$ and $\mu_Y$ that are $\mathcal{N}(0,I)$. For each $n$, we consider two alternative cases: $\left(\mu_X = \mathcal{N}(0, I), \mu_Y = \mathcal{N}(0.2, I)\right)$ and $\left(\mu_X = \mathcal{N}(0, I), \mu_Y = \mathcal{N}(0, 1.3I)\right)$.
The results are included in \Cref{fig:two_sample}. The distributions of $D_{H, n\mathrm{values}}(X,Y)$ in the null and alternative case are easily discernible. 


\begin{figure*}[t]
    \centering
    
    \begin{subfigure}[b]{\textwidth}
        \centering
        \begin{subfigure}[b]{0.4\textwidth}
            \centering
            \includegraphics[width=\textwidth]{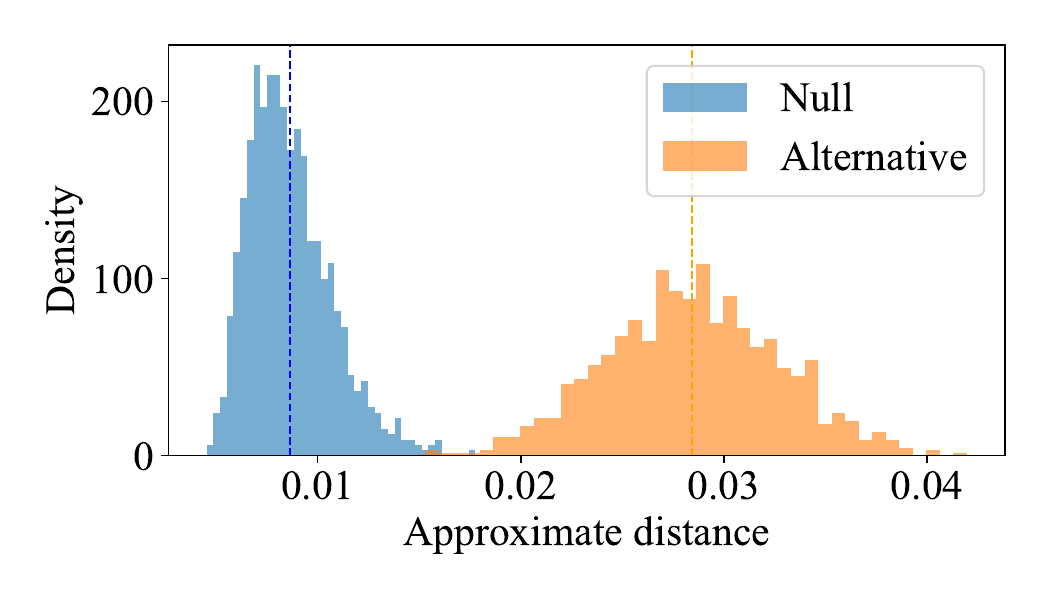}
            \caption*{Alternative $\left(\mu_X = \mathcal{N}(0, I), \mu_Y = \mathcal{N}(0.2, I)\right)$.}
        \end{subfigure}
        \hfill
        \begin{subfigure}[b]{0.4\textwidth}
            \centering
            \includegraphics[width=\textwidth]{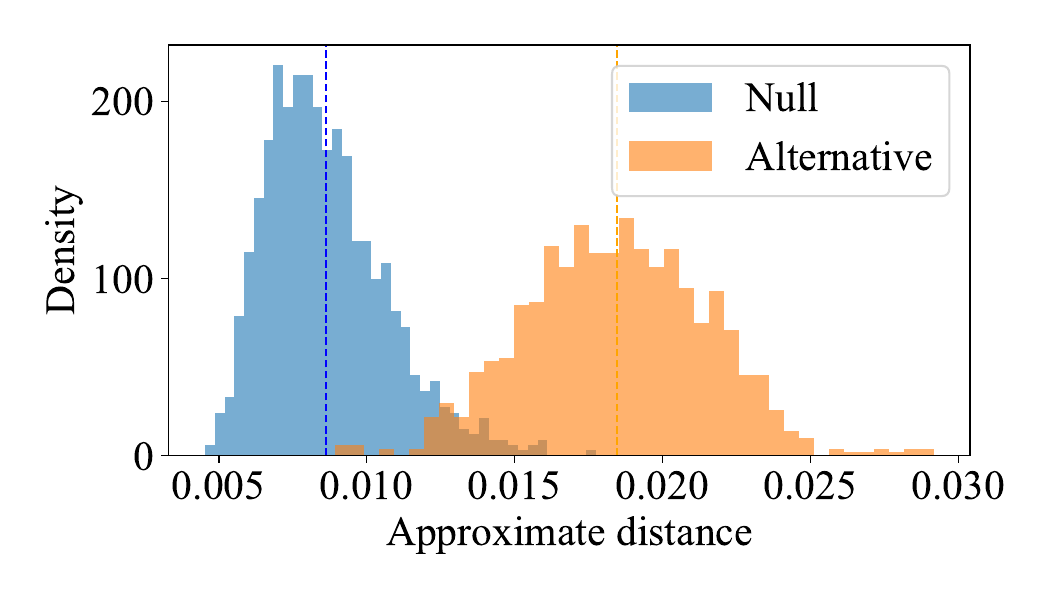}
            \caption*{Alternative $\left(\mu_X = \mathcal{N}(0, I), \mu_Y = \mathcal{N}(0, 1.3I)\right)$.}
        \end{subfigure}
        \caption{Dimensionality $n = 2$.}
        \label{fig:two_sample_n2}
    \end{subfigure}
    
    \vspace{0.5cm}
    
    \begin{subfigure}[b]{\textwidth}
        \centering
        \begin{subfigure}[b]{0.4\textwidth}
            \centering
            \includegraphics[width=\textwidth]{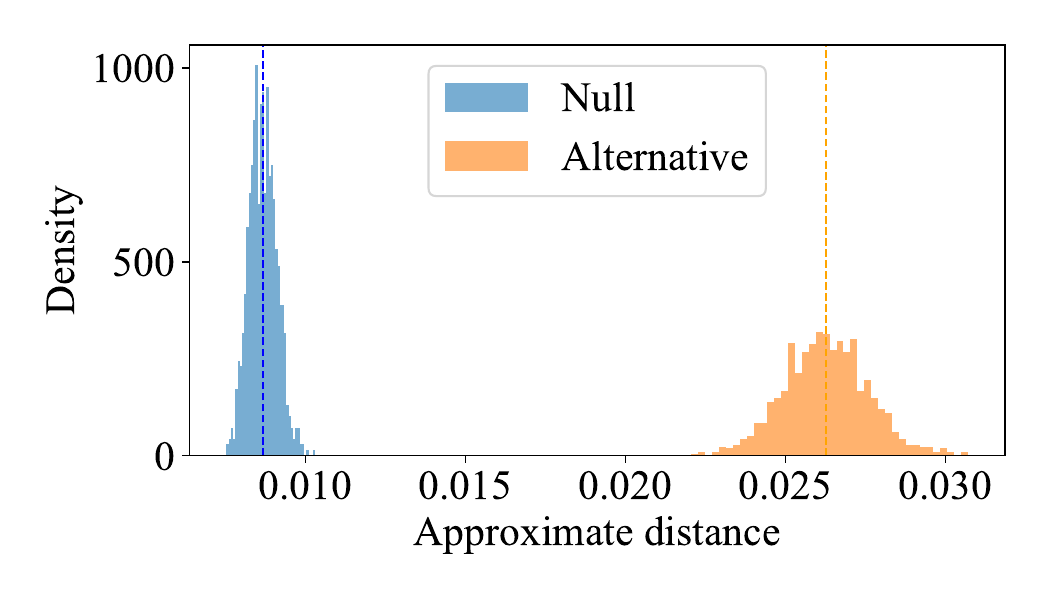}
            \caption*{Alternative $\left(\mu_X = \mathcal{N}(0, I), \mu_Y = \mathcal{N}(0.2, I)\right)$.}
        \end{subfigure}
        \hfill
        \begin{subfigure}[b]{0.4\textwidth}
            \centering
            \includegraphics[width=\textwidth]{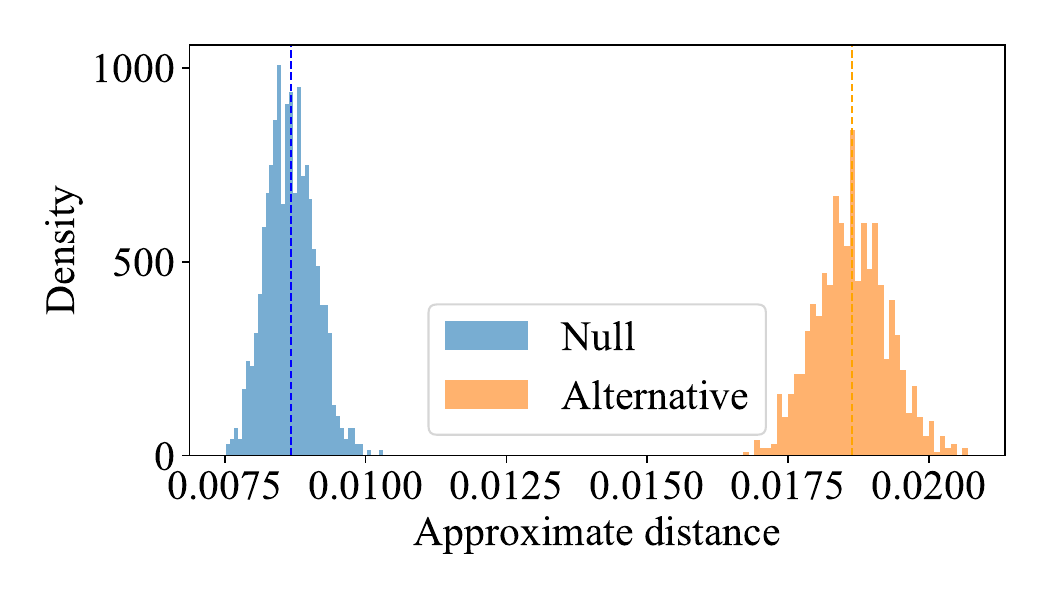}
            \caption*{Alternative $\left(\mu_X = \mathcal{N}(0, I), \mu_Y = \mathcal{N}(0, 1.3I)\right)$.}
        \end{subfigure}
        \caption{Dimensionality $n = 50$.}
        \label{fig:two_sample_n50}
    \end{subfigure}
    
    \caption{The empirical distribution of $D_{300, 100}(X,Y)$ for the null case and different alternative cases in the two-sample test. The randomness comes from the drawn samples for $X$ and $Y$ and the selection of half-spaces in \Cref{alg:dist_heur}. In all experiments, the null case is $\mu_X=\mu_Y=\mathcal{N}(0,I)$. Each experiment corresponds to a particular choice of the dimensionality $n$ and the alternative case. The distribution of $D_{300, 100}(X,Y)$ is given in blue for the null case and in orange for the alternative. The mean values are denoted by the vertical dotted lines.}
    \label{fig:two_sample}
\end{figure*}
\section{Distribution Steering}\label{sec:distr_steer}
In distribution steering, we consider a distribution for the state of a dynamical system at any time $t$ and the goal is to match the state distribution with a target distribution as $t \rightarrow \infty$.
We suppose the dynamical system 
\begin{equation}\label{eq:dyn}
    x_{t+1} = f(x_t, u_t),
\end{equation}
where $x_t \in \mathbb{R}^n$, $u_t \in \mathbb{R}^m$, and $f: \mathbb{R}^n \times \mathbb{R}^m \rightarrow \mathbb{R}^n$. We suppose that the state at the initial time $t=0$ follows a distribution $\mu_\mathrm{start}$, given as a set of samples $\{ x_i^{(0)} \}_{i=1}^N$. The distribution of the state at time $t$ is represented by the set of samples $\{ x_i^{(t)} \}_{i=1}^N$, in general. Our task is to design a closed-loop controller that steers the distribution of the state at $t$ towards the distribution $\mu_Y$, as $t\rightarrow \infty$. We choose the controller family $u_t = K_t x_t + b_t$, with parameters $(K_t, b_t)$. This is a common choice in the literature \cite{rapakoulias2024discrete, rapakoulias2023discrete}.

Our method for computing the controller at $t$, given the state distribution at $t$, $\{ x_i^{(t)} \}_{i=1}^N$, is as follows. 
Suppose we start with the controller $\pi_t=(K_t, b_t)$. Then, the samples at the next time, $t+1$, will be $\lbrace x^{(t+1)}_i\rbrace_{i=1}^N$, where
\begin{equation}
    x_i^{(t+1)} = f\left(x_i^{(t)}, \pi_t(x_i^{(t)})\right),\quad \pi_t(x_i^{(t)})=K_t x_i^{(t)} + b_t.
\end{equation}
Via a recursion we can obtain the samples after $\tau$ timesteps, $\{ x_i^{(t+\tau)}(K_t, b_t)\}_i$, where we make the dependence on $K_t$ and $b_t$ explicit. By applying \Cref{alg:dist_heur}, we can obtain the distance estimate $D_{H, n\mathrm{values}}\left(\lbrace x_i^{(t+\tau)}(K_t, b_t) \rbrace_i, \mu_Y\right)$, which is a differentiable function of the controller $(K_t, b_t)$. Therefore, to select the controller for time $t$, we start with an initial controller and apply gradient descent with step-size $\rho$ to improve our controller design:
\begin{equation}\label{eq:baseline}
    \pi_t \leftarrow \pi_t - \rho \nabla_{K_t, b_t}D_{H, n\mathrm{values}}\left(\lbrace x_i^{(t+\tau)}(K_t, b_t) \rbrace_i, \mu_Y\right).
\end{equation}
Although the method to compute the controllers is presented in \cref{eq:baseline} using simple gradient descent, we can apply more advanced descent algorithms, such as RMSProp, Adagrad, and Adam. We can use multiple restarts and other empirical techniques to improve the robustness and convergence of the optimization process. The ability to apply any of these methods however relies on the fact that our distance metric in the space of distributions is an easily computable and differentiable function of the controller.

\revboyd{We run variants of our algorithm \eqref{eq:baseline}, where we vary $H$ and $n_\mathrm{values}$. However, we judge the performance of all variants by looking at the estimate $D_{300, 100}$ for the distance between the state distribution and the target over time. We use $N=3000$.} For each $q_k$, we use equally spaced values between the $0.5$th and $99.5$th percentiles of $q_k^\top Y$, where $Y\sim \mu_Y$, for the $b_j$s. We set $\tau=25$ and only update the controller every $20$ steps.

We evaluate our algorithm on the unicycle dynamics model, which is nonlinear and given by the equation
\begin{equation}
    \begin{bmatrix}
        \chi_{t+1} \\
        \psi_{t+1}\\
        \theta_{t+1}
    \end{bmatrix} = \underbrace{\begin{bmatrix}
        \chi_t + u_{1,t} \Delta t \cos(\theta_t)\\
        \psi_t + u_{1,t} \Delta t \sin(\theta_t) \\
        \theta_t + u_{2,t}\Delta t
    \end{bmatrix}}_f,
\end{equation}
where $u_{1,t}$ is the linear velocity input and $u_{2,t}$ is the angular velocity input. We will only look at the distribution of the Cartesian position $(\chi, \psi)$. We use $\Delta t=0.1$. We will assume zero initial heading and
\begin{equation}
    \mu_\mathrm{start} = \mathcal{N}\left(\begin{bmatrix}
    -2\\
    -2
\end{bmatrix}, I\right),\ \mu_Y = \mathcal{N}\left(\begin{bmatrix}
    3\\
    2
\end{bmatrix}, \begin{bmatrix}
    2 &1.5\\
    1.5 &2
\end{bmatrix}\right)
\end{equation}
for the distributions of the position. We do not constrain the distribution of the heading at the final time and the initial heading is $0$.

Our results are summarized in \Cref{fig:distr_steering}. We observe that our algorithm \eqref{eq:baseline} is sufficient to steer the state distribution of the unicycle model. We cannot exactly reach the target distribution, because the unicycle dynamics are non-linear, the controller is affine, and gradient descent is not guaranteed to converge to the global optimum. \revboyd{The oscillations for some variants in \Cref{fig:distr_steer_3} are due to the unicycle dynamics and the approximation error of the distance estimate for these variants in \eqref{eq:baseline}.}

\begin{figure*}[ht]
    \centering
    \begin{subfigure}[t]{0.34\textwidth}
        \centering
        \includegraphics[width=\textwidth]{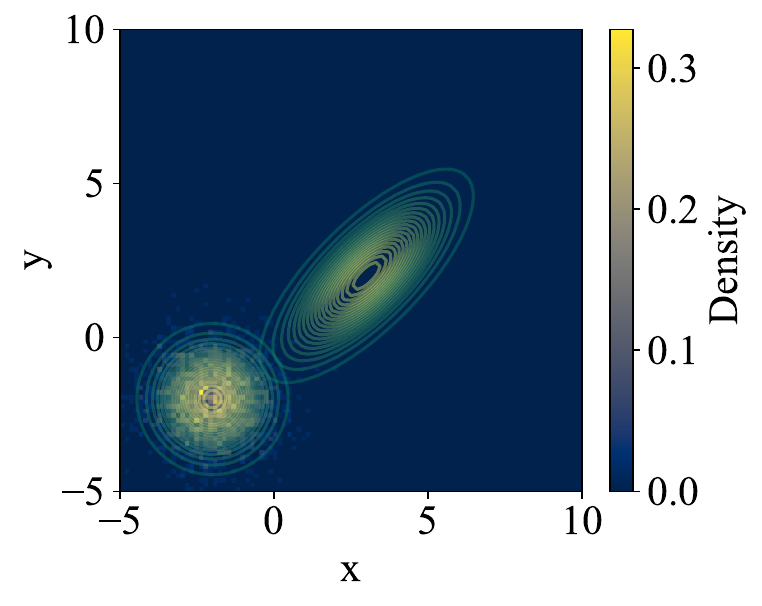}
        \caption{The sample density for the distribution of the initial position. The contours of the initial position distribution, $\mu_\mathrm{start}$, and the target distribution, $\mu_Y$ are also included.}
        \label{fig:distr_steer_1}
    \end{subfigure}
    \begin{subfigure}[t]{0.34\textwidth}
        \centering
        \includegraphics[width=\textwidth]{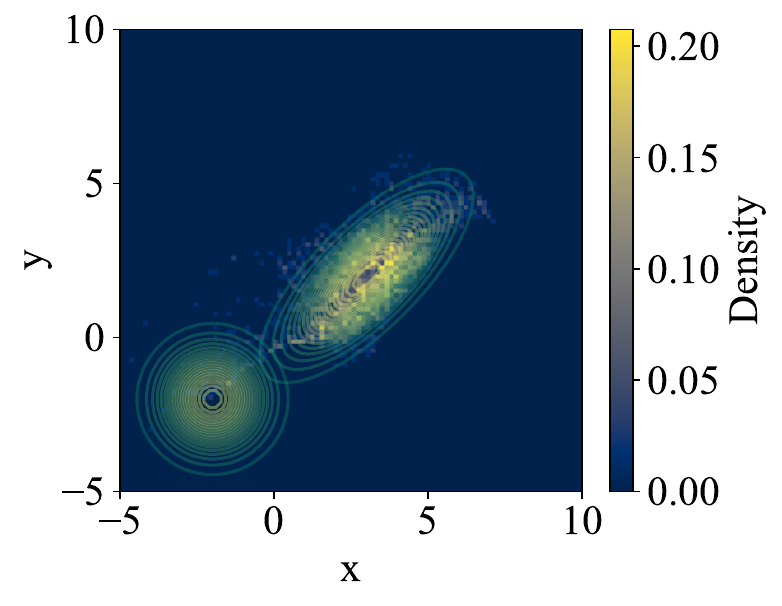} 
        \caption{The sample density for the distribution of the final position for the algorithm variant of \eqref{eq:baseline} with $H=200$ and $n_\mathrm{values}=400$. 
        Our proposed method steers the position distribution onto the target distribution.}
        \label{fig:distr_steer_2}
    \end{subfigure}
    \begin{subfigure}[t]{0.3\textwidth}
        \centering
        \includegraphics[width=\textwidth]{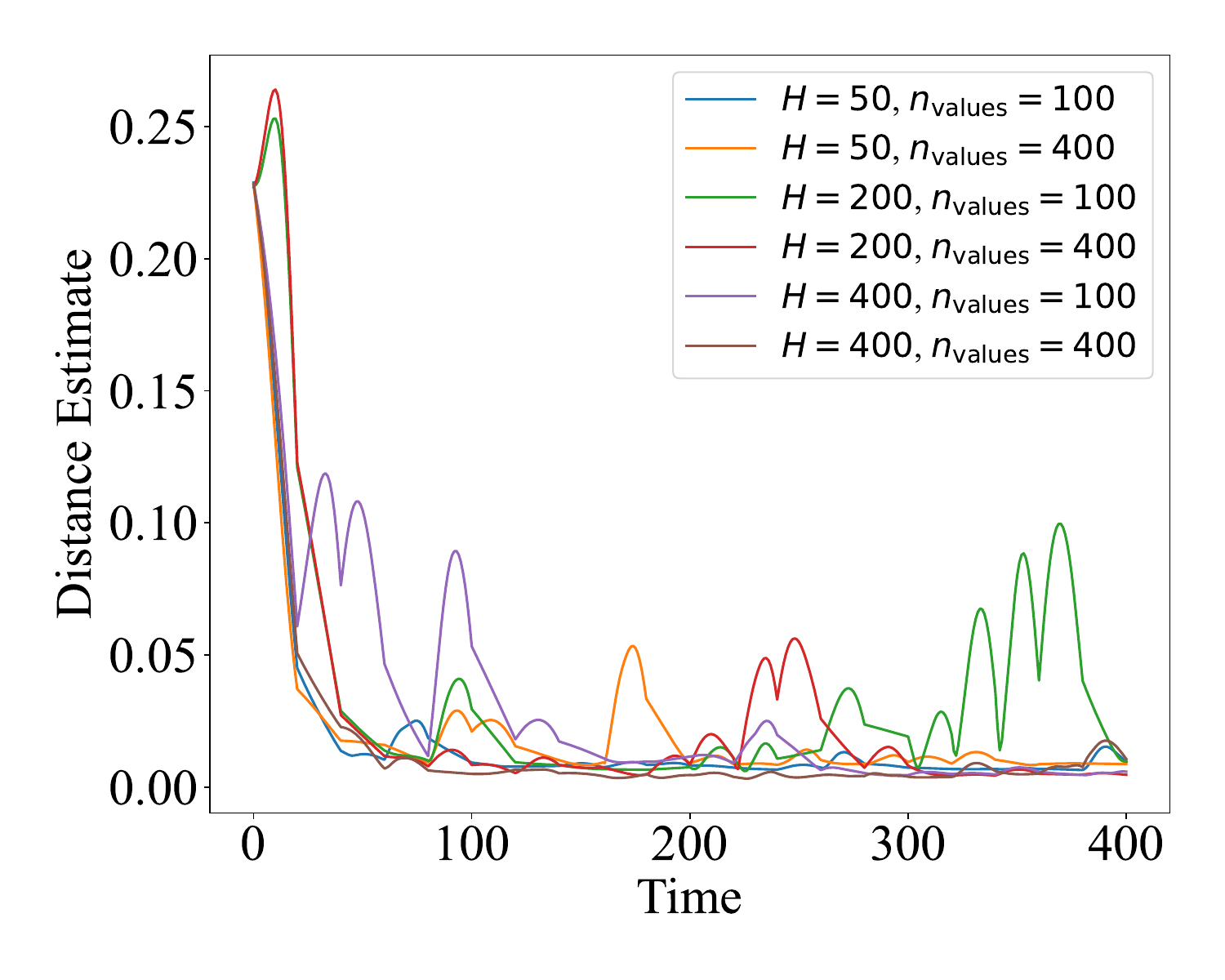}
        \caption{Running distance estimate $D_{300, 100}\left(\cdot, \mu_Y\right)$ between the sample distribution of positions for state $t$ and the target position distribution, for different variants of algorithm \eqref{eq:baseline}.}
        \label{fig:distr_steer_3}
    \end{subfigure}

    \caption{Results for the distribution steering application.}
    \label{fig:distr_steering}
\end{figure*}
\section{Ergodic Control}\label{sec:ergodic}

The goal of ergodic control is to control the empirical distribution of a dynamical system's state as given by its trajectory over time \cite{miller2013trajectory, dressel2019tutorial}. This is different than the goal of distribution steering, where we control the distribution of the dynamical system's state at any given time. 
We formulate the ergodic control problem as follows. We consider the dynamics \eqref{eq:dyn} and select the control sequence $u =: (u_0, \dots, u_T)$, such that the distance 
\begin{equation}
    D_{H, n\mathrm{values}}\left(\{x_t \}_{t=0}^{T+1}, \mu_Y\right)
\end{equation}
is minimized, where $x_{t+1} = f(x_t, u_t)$. Again, we can solve this problem using gradient descent
\begin{equation}\label{eq:gd_ergodic}
\begin{aligned}
    u \leftarrow u -
    \rho \nabla_{u} D_{H, n_\mathrm{values}}\left(\{x_t \}_{t=0}^{T+1}, \mu_Y\right)
    \end{aligned}
\end{equation}
or a more advanced descent algorithm.

\revboyd{We assume the dynamics are unicycle, the initial state is $x_0=(0,0,0)$, $T=5,000$, and the target position distribution $\mu_Y$ is a GMM of two components. We do not constrain the distribution of the heading trajectory. We run variants of \eqref{eq:gd_ergodic} with different $H$ and $n_\mathrm{values}$. However, we evaluate the performance of all variants using $D_{400, 400}$.} We use AdamW \cite{loshchilov2019decoupledweightdecayregularization}, instead of the simple gradient descent update in \eqref{eq:gd_ergodic}. For all $D_{H, n_\mathrm{values}}(\cdot, \mu_Y)$, for each $q_k$, we use equally spaced values between the $0.5$th and $99.5$th percentiles of $q_k^\top Y$, where $Y\sim \mu_Y$, for the $b_j$s. 

\begin{figure*}[t!]
    \centering
    \begin{subfigure}[t]{0.69\textwidth}
        \centering
        \includegraphics[width=\textwidth]{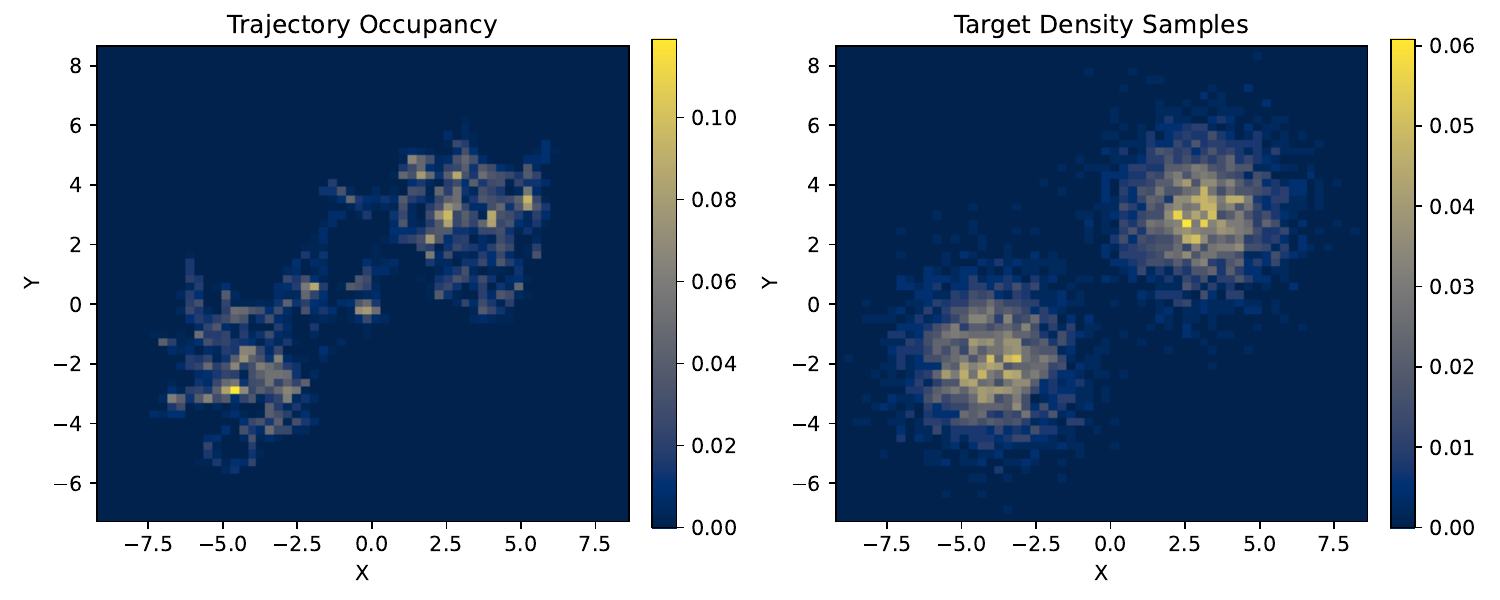}
        \caption{The state trajectory for the position, obtained by our algorithm \eqref{eq:gd_ergodic} and represented as a sample density on the left, and the target GMM position density on the right. We observe that the empirical distribution of the position trajectory is close to the target distribution. On the left, the trajectory was obtained by the variant of algorithm \eqref{eq:gd_ergodic} with $H=50$ and $n_\mathrm{values}=400$. All other variants also get close to the target distribution, as indicated by \Cref{fig:ergodic_2}. }
        \label{fig:ergodic_1}
    \end{subfigure}
    \hfill
    \begin{subfigure}[t]{0.3\textwidth}
        \centering
        \includegraphics[width=\textwidth]{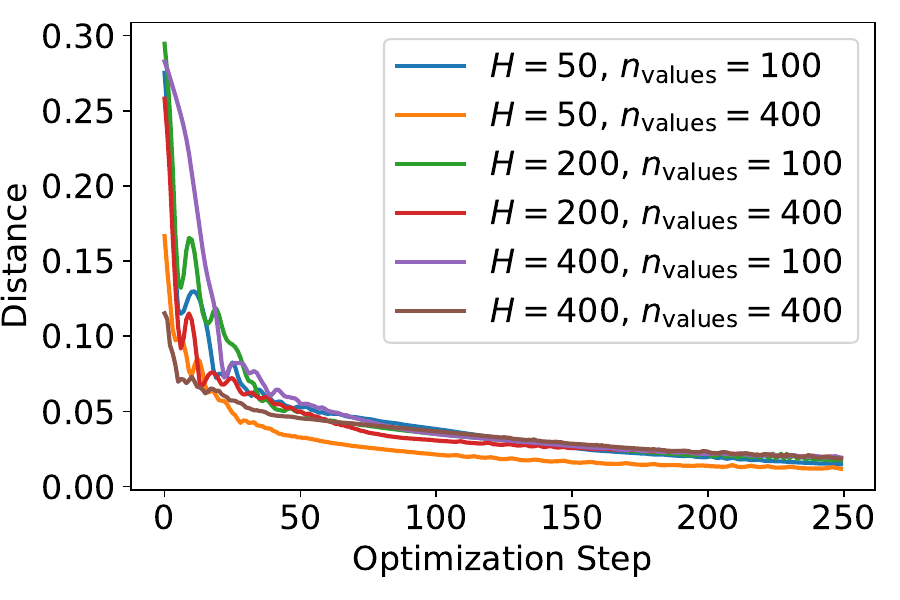} 
        \caption{Distance estimate $D_{400, 400}\left(\{x_t \}_{t=0}^{T+1}, \mu_Y\right)$ of the the position trajectory to the target distribution as a function of optimization step for variants of algorithm \eqref{eq:gd_ergodic}.}
        \label{fig:ergodic_2}
    \end{subfigure}

    \caption{Results for the ergodic control application.}
    \label{fig:ergodic}
\end{figure*}

Our results are summarized in \Cref{fig:ergodic}. If we look at the position trajectory $x_0, \dots, x_{T+1}$ as a set of samples, we notice that the empirical distribution is close to the target distribution $\mu_Y$. This is also verified by the loss at every gradient descent step that decreases until it reaches a point close to $0$. The loss at every optimization step is the distance of the position trajectory at that step to the target distribution.

\section{Conclusion}\label{sec:conclusion}
We suggest a family of distances between distributions that is interpretable and admits an easy-to-compute and differentiable approximation. We further present proof-of-concept applications of the proposed distance to control tasks. Future work will focus on  the convergence of \Cref{alg:dist_heur} as a function of the problem dimensionality and the sensitivity of \Cref{alg:dist_heur} to the selection of half-spaces.




\section*{Acknowledgments}

\footnotesize
Toyota Research Institute (TRI) provided funds to assist the authors with their research, but this article solely reflects the opinions and conclusions of its authors and not TRI or any other Toyota entity. The NASA University Leadership Initiative (grant $\#$80NSSC20M0163) provided funds to assist the first author with their research, but this article solely reflects the opinions and conclusions of its authors and not any NASA entity. For the first author, this work was also partially funded through the Alexander S. Onassis Foundation Scholarship program. The authors thank Logan Bell and Emmanuel Candés for their feedback.

\renewcommand*{\bibfont}{\footnotesize}
\printbibliography

@article{balci2020covariance,
  title={Covariance {S}teering of {D}iscrete-{T}ime {S}tochastic {L}inear {S}ystems based on {W}asserstein {D}istance {T}erminal {C}ost},
  author={Balci, Isin M and Bakolas, Efstathios},
  journal={IEEE Control Systems Letters},
  volume={5},
number={6},
  year={2020},
  publisher={IEEE}
}

@book{kochenderfer2022algorithms,
  title={Algorithms for {D}ecision {M}aking},
  author={Kochenderfer, Mykel J and Wheeler, Tim A and Wray, Kyle H},
  year={2022},
  publisher={MIT {P}ress}
}

@book{bengio2017deep,
  title={Deep {L}earning},
  author={Bengio, Yoshua and Goodfellow, Ian and Courville, Aaron},
  year={2017},
  publisher={MIT Press}
}

@article{givens1984class,
  title={A {C}lass of {W}asserstein {M}etrics for {P}robability {D}istributions.},
  author={Givens, Clark R and Shortt, Rae Michael},
  journal={Michigan Mathematical Journal},
  volume={31},
  number={2},
  year={1984},
  publisher={University of Michigan, Department of Mathematics}
}

@article{gneiting2007strictly,
  title={Strictly {P}roper {S}coring {R}ules, {P}rediction, and {E}stimation},
  author={Gneiting, Tilmann and Raftery, Adrian E},
  journal={Journal of the {A}merican {S}tatistical {A}ssociation},
  volume={102},
  number={477},
  year={2007},
  publisher={Taylor \& Francis}
}

@book{billingsley2017probability,
  title={Probability and {M}easure},
  author={Billingsley, Patrick},
  year={2017},
  publisher={John Wiley \& Sons}
}

@article{kolouri2019generalized,
  title={Generalized {S}liced {W}asserstein {D}istances},
  author={Kolouri, Soheil and Nadjahi, Kimia and Simsekli, Umut and Badeau, Roland and Rohde, Gustavo},
  journal={Advances in {N}eural {I}nformation {P}rocessing {S}ystems},
  year={2019}
}

@inproceedings{titouan2019sliced,
  title={Sliced {G}romov-{W}asserstein},
  author={Titouan, Vayer and Flamary, R{\'e}mi and Courty, Nicolas and Tavenard, Romain and Chapel, Laetitia},
  booktitle={Advances in Neural Information Processing Systems},
  year={2019}
}

@inproceedings{lin2021projection,
  title={On {P}rojection {R}obust {O}ptimal {T}ransport: {S}ample {C}omplexity and {M}odel {M}isspecification},
  author={Lin, Tianyi and Zheng, Zeyu and Chen, Elynn and Cuturi, Marco and Jordan, Michael I},
  booktitle={International Conference on Artificial Intelligence and Statistics},
  year={2021}
}

@article{grossi2025refereeing,
  title={Refereeing the {R}eferees: {E}valuating {T}wo-{S}ample {T}ests for {V}alidating {G}enerators in {P}recision {S}ciences},
  author={Grossi, Samuele and Letizia, Marco and Torre, Riccardo},
  journal={Machine Learning: Science and Technology},
  volume={6},
  year={2025}
}

@inproceedings{deshpande2018generative,
  title={Generative {M}odeling using the {S}liced {W}asserstein {D}istance},
  author={Deshpande, Ishan and Zhang, Ziyu and Schwing, Alexander G},
  booktitle={IEEE {C}onference on {C}omputer {V}ision and {P}attern {R}ecognition},
  year={2018}
}

@inproceedings{rabin2011wasserstein,
  title={Wasserstein {B}arycenter and its {A}pplication to {T}exture {M}ixing},
  author={Rabin, Julien and Peyr{\'e}, Gabriel and Delon, Julie and Bernot, Marc},
  booktitle={International {C}onference on {S}cale {S}pace \& {V}ariational {M}ethods in {C}omputer {V}ision},
  year={2011}
}

@article{kim2020robust,
  title={Robust {M}ultivariate {N}onparametric {T}ests via {P}rojection {A}veraging},
  author={Kim, Ilmun and Balakrishnan, Sivaraman and Wasserman, Larry},
  journal={The Annals of Statistics},
  volume={48},
  number={6},
  year={2020},
  publisher={JSTOR}
}

@inproceedings{li2020projective,
  title={On a {P}rojective {E}nsemble {A}pproach to {T}wo {S}ample {T}est for {E}quality of {D}istributions},
  author={Li, Zhimei and Zhang, Yaowu},
  booktitle={International Conference on Machine Learning},
  year={2020},
}

@inproceedings{kolouri2018sliced,
  title={Sliced {W}asserstein {D}istance for {L}earning {G}aussian {M}ixture {M}odels},
  author={Kolouri, Soheil and Rohde, Gustavo K and Hoffmann, Heiko},
  booktitle={IEEE {C}onference on {C}omputer {V}ision and {P}attern {R}ecognition},
  year={2018}
}

@article{mathew2011metrics,
  title={Metrics for {E}rgodicity and {D}esign of {E}rgodic {D}ynamics for {M}ulti-{A}gent {S}ystems},
  author={Mathew, George and Mezi{\'c}, Igor},
  journal={Physica D: Nonlinear Phenomena},
  volume={240},
  number={4-5},
  year={2011},
  publisher={Elsevier}
}

@article{nadjahi2020statistical,
  title={Statistical and {T}opological {P}roperties of {S}liced {P}robability {D}ivergences},
  author={Nadjahi, Kimia and Durmus, Alain and Chizat, L{\'e}na{\"\i}c and Kolouri, Soheil and Shahrampour, Shahin and Simsekli, Umut},
  journal={Advances in Neural Information Processing Systems},
  year={2020}
}

@incollection{su2015distances,
  title={Distances and {K}ernels based on {C}umulative {D}istribution {F}unctions},
  author={Su, Hongjun and Zhang, Hong},
  booktitle={Emerging Trends in Image Processing, Computer Vision and Pattern Recognition},
  year={2015},
  publisher={Elsevier}
}

@article{tzikas2025distributionally,
  title={Distributionally {R}obust {C}ontrol with {C}onstraints on {L}inear {U}nidimensional {P}rojections},
  author={Tzikas, Alexandros E and Fiechtner, Lukas and Jamgochian, Arec and Kochenderfer, Mykel J},
  journal={arXiv preprint arXiv:2508.07121},
  year={2025}
}

@misc{loshchilov2019decoupledweightdecayregularization,
      title={Decoupled {W}eight {D}ecay {R}egularization}, 
      author={Ilya Loshchilov and Frank Hutter},
      year={2019},
      eprint={1711.05101},
      archivePrefix={arXiv},
}

@article{araya2010pomdp,
  title={A {POMDP} {E}xtension with {B}elief-{D}ependent {R}ewards},
  author={Araya, Mauricio and Buffet, Olivier and Thomas, Vincent and Charpillet, Fran{\c{c}}cois},
  journal={Advances in Neural Information Processing Systems},
  year={2010}
}

@inproceedings{rapakoulias2024discrete,
  title={Discrete-{T}ime {M}aximum {L}ikelihood {N}eural {D}istribution {S}teering},
  author={Rapakoulias, George and Tsiotras, Panagiotis},
  booktitle={IEEE Conference on Decision and Control (CDC)},
  year={2024},
  organization={IEEE}
}

@inproceedings{rapakoulias2023discrete,
  title={Discrete-{T}ime {O}ptimal {C}ovariance {S}teering via {S}emidefinite {P}rogramming},
  author={Rapakoulias, George and Tsiotras, Panagiotis},
  booktitle={IEEE Conference on Decision and Control (CDC)},
  year={2023},
  organization={IEEE}
}

@inproceedings{miller2013trajectory,
  title={Trajectory {O}ptimization for {C}ontinuous {E}rgodic {E}xploration},
  author={Miller, Lauren M and Murphey, Todd D},
  booktitle={American Control Conference},
  year={2013},
  organization={IEEE}
}

@book{measure,
publisher={Wiley-Interscience},
author={Kramer, Walter},
title={Probability  {M}easure},
year={1995},
}

@article{dressel2019tutorial,
  title={Tutorial on the {G}eneration of {E}rgodic {T}rajectories with {P}rojection-based {G}radient {D}escent},
  author={Dressel, Louis and Kochenderfer, Mykel J},
  journal={IET Cyber-Physical Systems: Theory \& Applications},
  volume={4},
  number={2},
  year={2019},
  publisher={Wiley Online Library}
}

\end{document}